\documentclass[letterpaper, 10 pt, conference]{ieeeconf}  

\IEEEoverridecommandlockouts                              

\usepackage{amsmath} 
\usepackage{graphicx}
\usepackage{siunitx}
\usepackage{overpic}
\usepackage{rotating}
\usepackage{tabularx}
\usepackage{underscore}
\usepackage{hyperref}
\hypersetup{hidelinks,
	colorlinks=true,
	allcolors=black,
	pdfstartview=Fit,
	breaklinks=true}
\usepackage{cite}
\usepackage{amssymb}  
\usepackage{xcolor}
\usepackage{algorithm}
\usepackage{algpseudocode}

\renewcommand{\dblfloatsep}{0pt} 
\renewcommand{\dbltextfloatsep}{0pt} 
\renewcommand{\intextsep}{0pt} 

\allowdisplaybreaks 

\def\aujour{\number\day \space \ifcase\month\or
janvier\or f�vrier\or mars\or avril\or mai\or
juin\or juillet\or ao�t\or septembre\or octobre\or
novembre\or d�cembre\fi \space \number\year}

\def\cH{{\cal H}}

\def\cL{{\cal L}}

\newtheorem{remark}{Remark}
\newtheorem{ass}{Assumption}

\newtheorem{theorem}{Theorem}

\def\C{{\setbox0=\hbox{$\displaystyle{\rm C}$}
        \hbox{\hbox to0pt{\kern 0.4\wd0\vrule height 0.95\ht0\hss}\box0}}}
\def\Q{{\setbox0=\hbox{$\displaystyle{\rm Q}$}%
    \hbox{\raise 0.2\ht0\hbox to0pt{\kern 0.4\wd0\vrule height
    0.85\ht0\hss}\box0}}} 

\def\cH2{{\cal H}_2} 

\def\cL2{\mathop{\mathcal L}_{2}} 

\def\cRH2{\mathop{\cal R \cal H}_2} 
\def\cRL2{\mathop{\cal R \cal L}_{2}} 

\newcommand{\norm}[1]{\left\|{#1}\right\|}

\makeatletter
\DeclareRobustCommand\sfrac[1]{\@ifnextchar/{\@sfrac{#1}}
                                            {\@sfrac{#1}/}}
\def\@sfrac#1/#2{\leavevmode\kern.1em\raise.5ex
         \hbox{$\m@th\fontsize\sf@size\z@
                           \selectfont#1$}\kern-.1em
         /\kern-.15em\lower.25ex
          \hbox{$\m@th\fontsize\sf@size\z@
                            \selectfont#2$}}
\makeatother

\title{\LARGE \bf
Unsupervised Physics-Informed Neural Network-based Nonlinear Observer design for autonomous systems using contraction analysis}

\author{Yasmine Marani\textsuperscript{1,*}, Israel Filho\textsuperscript{1,*}, Tareq Al-Naffouri\textsuperscript{1} and Taous-Meriem Laleg-Kirati\textsuperscript{2}\IEEEmembership{IEEE Senior Members}
\thanks {\textsuperscript{1} Computer, Electrical and Mathematical Science and Engineering Division (CEMSE), King Abdullah University of Science and Technology (KAUST), Thuwal 23955-6900, Saudi Arabia
\thanks \textsuperscript{2} Taous Meriem Laleg-Kirati is affiliated with the National Institute for Research in Digital Science and Technology, Paris-Saclay, France.
\thanks \textsuperscript{*} equal contributors. 
        {\tt\small yasmine.marani@kaust.edu.sa,
        israel.filho@kaust.edu.sa,
        tareq.alnaffouri@kaust.edu.sa,
        taous-meriem.laleg@inria.fr }}%
}

\begin{document}
\maketitle
\sloppy
\thispagestyle{empty}
\pagestyle{empty}
\begin{abstract}
Contraction analysis offers, through elegant mathematical developments, a unified way of designing observers for a general class of nonlinear systems, where the observer correction term is obtained by solving an infinite dimensional inequality that guarantees global exponential convergence. However, solving the matrix partial differential inequality involved in contraction analysis design is both analytically and numerically challenging and represents a long-lasting challenge that prevented its wide use. Therefore, the present paper proposes a novel approach that relies on an unsupervised  Physics Informed Neural Network (PINN) to design the observer's correction term by enforcing the partial differential inequality in the loss function. The performance of the proposed PINN-based nonlinear observer is assessed in numerical simulation as well as its robustness to measurement noise and neural network approximation error.  
\end{abstract}

\section{INTRODUCTION}
\label{sec:intro}
\noindent 

Nonlinear observer design is a fundamental research area in control theory that is constantly attracting attention from researchers in the community.  While general and systematic frameworks for state estimation of linear systems with global convergence guarantees are well-established in the literature \cite{Luenberger1964,Kalman1960}, nonlinear observer design still suffers from a lack of generality and global convergence guarantees.\\
The literature abounds with various nonlinear observer methods such as high gain observers \cite{khalil2008}, immersion and invariance-based observers \cite{karagiannis2008invariant}, observers based on geometric methods \cite{KRENER1983}, those based on Linear Matrix Inequalities (LMI)s \cite{Rajamani1998}, algebraic estimators \cite{Yasmine2023}, approaches relying on an injective transformation into a larger latent space \cite{Kazantzis1997}, and the well-known Extended Kalman Filter \cite{anderson1979}. However, most of the above-mentioned observer designs rely heavily on the class of nonlinearity of the system or provide only local convergence guarantees. On the other hand, observer design based on contraction analysis is slowly regaining popularity since its first introduction in the late nineties in \cite{slotine1996a}. They offer a generic design framework that is suitable for a general class of smooth nonlinear systems while providing global exponential convergence guarantees. \\
Contraction analysis was introduced in \cite{LOHMILLER1998683}, drawing on concepts from continuum mechanics and differential geometry. It offers a distinct perspective on stability analysis compared to traditional methods like Lyapunov theory \cite{Nijmeijer1990, isidori1995nonlinear,  khalil2002nonlinear}. Rather than focusing on the convergence of a system to a target trajectory or equilibrium point, contraction analysis assesses stability by determining whether all trajectories of a system converge toward one another. In essence, a system is considered contracting if it eventually "forgets" its initial conditions or any temporary disturbances, making this approach particularly well-suited for observer design. Indeed, the initial motivation behind contraction analysis was primarily related to the design of observers for nonlinear systems \cite{manchester2018contracting}. Numerous contraction-based nonlinear controllers and observers have since been proposed in the literature such as the work in \cite{slotine1996a,slotine1996b, lohmiller1997,  LOHMILLER1998683, sanfelice2011convergence, dani2014observer, manchester2018contracting}. Contraction analysis provides a systematic framework for designing nonlinear observers, leveraging sophisticated mathematical developments that ensure global exponential convergence. The observer's correction term is determined by solving a Matrix Partial Differential Inequality (MPDI). However, despite the theoretical elegance of this approach, solving the MPDI presents significant analytical and numerical challenges that limit the use and implementation of contraction-based observers in their original form \cite{PeterGiesl2023}.
Therefore, we aim in the present paper to overcome this challenge by proposing a learning-based approach to design the contraction-based nonlinear observer's gain satisfying the MPDI.  \\
\noindent 
Artificial Neural Networks (ANNs) have emerged as powerful tools in approximating solutions to ordinary and partial differential equations (ODEs and PDEs). By leveraging their ability to approximate complex nonlinear functions, neural networks can be trained to satisfy the differential equations governing physical phenomena directly. This approach, often referred to as Physics-Informed Neural Networks (PINNs), was pioneered by Raissi \textit{et al.}~\cite{raissi2019}, and incorporates the underlying physical laws into the learning process, allowing the network to learn solutions that are consistent with the known physics. The advent of automatic differentiation has further revolutionized this field by enabling the efficient computation of derivatives required in the training process. This capability facilitates the direct incorporation of differential equations into the neural network's loss function, ensuring that the learned solutions not only fit the data but also satisfy the physical laws described by the differential equations. \\
\noindent
Physics-Informed Neural Networks (PINNs) have been extensively applied in modeling and parameter estimation of nonlinear dynamical systems. For instance, in \cite{zhai2023parameter}, a Runge–Kutta-based PINN framework was proposed for parameter estimation and modeling of nonlinear systems where the physical loss improves the integral involved in the Runge-Kutta method, reducing the error in the learnable trajectories. From the perspective of this work, PINNs have been utilized in observer design for discrete-time and continuous-time nonlinear systems. The work in \cite{alvarez2024nonlinear} introduced a PINN-based approach for designing nonlinear discrete-time observers, showing improved performance in state estimation without the need for explicit system models. Moreover, learning-based methods have been employed to design the KKL observer for autonomous nonlinear systems, as presented in \cite{niazi2023learning, peralez2021deep}, and nonlinear systems subject to measurement delay \cite{Yasmine2023kkl, MARANI2025}, where neural networks were trained to approximate a forward and inverse map involved in KKL observer design. The work in \cite{antonelo2024physics, ekeland2024physics} has shown that PINNs are effective in model predictive control, addressing practical challenges in the oil and gas industry, and accurately identified nonlinear dynamical systems. These findings emphasize the versatility of PINNs in system identification and observer design by embedding physical laws directly within neural network architectures, resulting in models that are both data-efficient and physically consistent, achieving robust generalization even for time horizons far beyond the training or operating points and maintaining resilience against additive noise.


\noindent
In the present paper, we propose an unsupervised learning approach to design a nonlinear observer for a general class of nonlinear systems based on the contraction analysis. The proposed approach relies on a Physics Informed Neural Network to enforce the contraction condition in the learning process. Based on the MDPI involved in the gain design of contraction-based nonlinear observer, we formulate an optimization problem by strategically designing the cost function of the Physics Informed Neural Network. Furthermore, we establish the robustness of the proposed learning-based nonlinear observer to the neural network approximation error and measurement noise and derive conditions ensuring exponential input-to-state stability. \\
The present paper is organized as follows:  Section \ref{sec:prelim} gives some background on contraction analysis and contraction-based nonlinear observers. The problem addressed in this paper is formulated in section \ref{sec:formulation}. Subsequently, we present the proposed unsupervised learning approach to design the observer's gain in section \ref{sec:pinnDes}, and analyze the robustness of the designed observer to the neural network approximation error and measurement noise in section \ref{sec:robustness}. The performance and robustness of the proposed observer is evaluated through two numerical examples in \ref{sec:sim}. Finally, a summary of the contributions and future work directions are provided in section \ref{sec:conclusion}.   
\textbf{Notation.} \\
For a square matrix $M$, $\operatorname{He}\{M\}=\frac{1}{2}(M+M^T)$ is the Hermitian Part of the matrix $M$. The class $\mathcal{C}^1$ is the class of continuous functions with continuous first derivatives. The Euclidean norm of a vector $u$ is denoted by $\norm{u}$. 


\section{Preliminaries on Contraction Theory}\label{sec:prelim}
Inspired by fluid mechanics and differential geometry,  \textit{Contraction Analysis} offers an alternative way of studying stability. Usually, stability is studied with respect to a nominal trajectory or an equilibrium point. Instead, contraction analysis studies the convergence of the solutions of a dynamical system to each other. In other words, the system is considered stable if the system's final behavior is independent of the initial condition \cite{LOHMILLER1998683}. A central result that was derived in \cite{slotine1996a,slotine1996b} is that if all neighboring trajectories converge to each other, then the system is globally exponentially stable and all trajectories converge to a single one.     \\
We consider the following general autonomous nonlinear system
\begin{equation}\label{model}
    \left\{\begin{array}{l}\dot{x}(t)=f(x)\\
    y(t)=h(x),
    \end{array}\right.
\end{equation}
where $x(t) \in \mathbb{R}^n$ is the state,  $y(t) \in \mathbb{R}^p$ is the output,  $f: \mathbb{R}^n  \rightarrow \mathbb{R}^n$ and $h: \mathbb{R}^n \rightarrow \mathbb{R}^p$ are smooth vector fields. 
The main result of \cite{LOHMILLER1998683} is that if there exists a bounded symmetric and positive definite matrix $P \in \mathbb{R}^{n \times n}$, such that 

\begin{equation}
\left(\frac{\partial f^{\mathrm{T}}}{\partial x} P+P \frac{\partial f}{\partial x}\right) \leq-2\lambda P.
\label{contract}
\end{equation}

Then, system \eqref{model} is contracting with rate $\lambda$. Furthermore, this result is considered to be a generalization and strengthening of the Krasovskii's global convergence theorem \cite{hahn1967}.

\subsection{Contraction analysis in the context of nonlinear observer design}
One of the main features of contraction analysis is the ability to conclude about a system's stability independently of the knowledge of any predefined nominal trajectory, which is a particularly attractive feature for observer design. Contraction analysis offers, therefore, a universal way of designing observers for nonlinear systems as in \eqref{model}, provided that the system is differentially detectable, which is a necessary condition for the existence of an exponentially stable observer of the form \cite{BERNARD2022}: 
\begin{equation}
\label{observer}
\left\{\begin{array}{l}
\dot{\hat{x}}=f(\hat{x})+k(\hat{x}, y) \\
\hat{y}=h(\hat{x}),
\end{array}\right.
\end{equation}
where $\hat{x} \in \mathbb{R}^n$ is the estimated state,  $\hat{y} \in \mathbb{R}^p$ is the estimated output, and $k: \mathbb{R}^n \times \mathbb{R}^p \rightarrow \mathbb{R}^n $ is the observer's nonlinear correction term. The correction term $k$ and the nonlinear function $f$ are assumed to be of at least class $\mathcal{C}^1$. \\
The following theorem provides an approach to designing the correction term $k$ based on the contraction analysis. 

\begin{theorem} \cite{BERNARD2022}
\label{th1}
Consider the smooth nonlinear system \eqref{model}. If there exists a positive definite matrix $P\in \mathbb{R}^{n\times n}$, a $\mathcal{C}^1$ function $k: \mathbb{R}^n \times \mathbb{R}^p \rightarrow \mathbb{R}^n $ and a real positive number $\lambda$, such that 
\begin{equation}\label{contractth}
\left\{\begin{aligned}
& \operatorname{He}\left\{P\left[\frac{\partial f}{\partial \hat{x}}(\hat{x})+\frac{\partial k}{\partial \hat{x}}(\hat{x}, y)\right]\right\} \leqslant-2\lambda P
 \quad \forall(\hat{x}, y) \in \mathbb{R}^n \times \mathbb{R}^p\\
 &k(x, h(x))=0 \quad \forall x \in \mathbb{R}^n,  \\
\end{aligned} \right.
\end{equation}
Then the observer in \eqref{observer} is a globally exponentially stable observer for system \eqref{model}. 
\end{theorem}

\section{Problem Formulation }\label{sec:formulation}
While Theorem \ref{th1} provides a general approach for designing nonlinear observers, determining the correction term involves solving a matrix partial differential inequality, which is highly challenging both analytically and numerically, and to the best of our knowledge, no existing numerical solvers are specifically equipped to handle this class of problems. To overcome this limitation, we exploit the universal approximation principle of neural networks to numerically approximate the observer's correction term. We specifically rely on an unsupervised Physics-Informed Neural Network approach that enforces the contraction conditions of Theorem \ref{th1} to learn the observer's gain. \\
Moreover, although the observer in \eqref{observer} satisfying \eqref{contractth} is global, using a neural network requires training on a closed set of interests. Therefore, we consider in the remainder of the paper, systems in the form of \eqref{model} that satisfy the following assumption

\begin{ass}
\label{assumption1}
    System \eqref{model} is forward invariant within $\mathcal{X}$ i.e.  there exist a compact set $\mathcal{X} \subset \mathbb{R}^n$, such that for all initial conditions $x(0)\in \mathcal{X}_0 \subset \mathcal{X}$ and all $t >0, X\left(t,x_0\right) \in \mathcal{X}$ and $Y\left(t,x_0\right) \in \mathcal{Y}$. Where $X\left(t,x_0\right)$ is the solution of \eqref{model} at time $t$, initialized at $x(0)=x_0$, and $Y\left(t,x_0\right)$ the corresponding output.
\end{ass}

\begin{remark}
    To reduce the complexity of the MDPI in \eqref{contractth}, we fix the matrix $P$ and solve \eqref{contractth} for $k$, instead of solving for both $P$ and $k$. Therefore, we consider in the rest of the paper $P=I$, where $I$ is the $n \times n$ identity matrix.   
\end{remark}


\section{Unsupervised learning of the observer's correction term }\label{sec:pinnDes}

\begin{figure}
    \centering
    \includegraphics[width=1\linewidth]{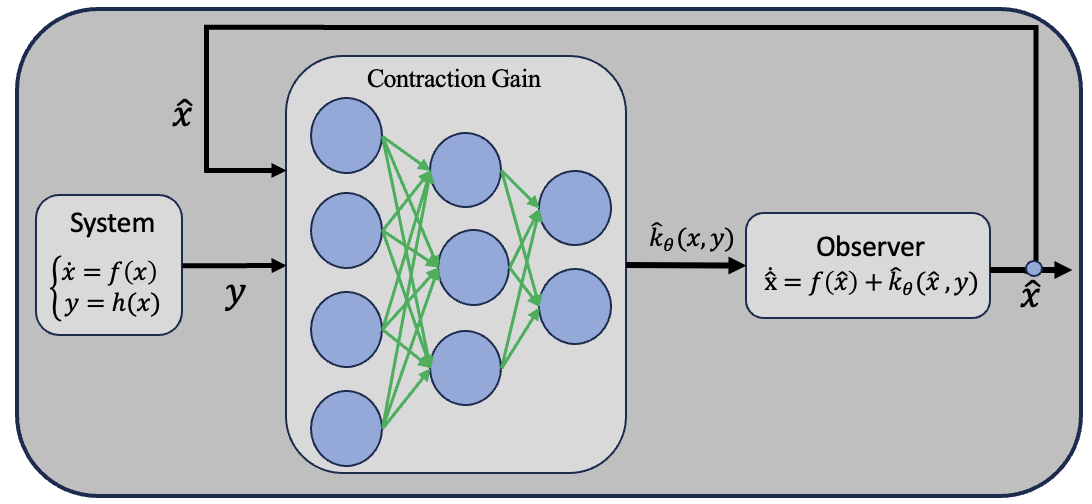}
    \caption{Block diagram of the proposed unsupervised learning-based contraction nonlinear observer design.}
    \label{fig:systemVisual}
\end{figure}



In the present section, we present the proposed Physics-Informed Neural Network approach to learn the gain of the observer in \eqref{observer}.  
Since the goal is to learn the correction term of the observer, we do not have pre-calculated trajectories for the estimated state $\hat{x}$. Therefore, we propose an unsupervised PINN-based approach to learn the observer's gain by leveraging the contraction analysis and enforcing the conditions in Theorem \ref{th1} to guarantee the exponential stability of the learned observer.

\noindent
Let $\hat{k}_\theta: \mathbb{R}^n \times \mathbb{R}^p \rightarrow \mathbb{R}^n $
be the learned observer gain given by 
\begin{equation}
\label{eq:mlpGeneral}
\hat{k}_{\theta}(\hat{x}, y) = T_{\theta}(\hat{x}, y),
\end{equation}

\noindent
where \(T_{\theta}\) is the function parameterized by the neural networks and $\theta$ represents the weights and biases. The block diagram of the proposed learning-based nonlinear observer is depicted in Fig. \ref{fig:systemVisual}. 
By integrating this contraction condition into the neural network training, we effectively regularize the model using the physics knowledge and constraints of the observer, which allows the neural network to learn the observer's correction term solely from the contraction condition, reducing, therefore, the need for large datasets. Therefore, we formulate the contraction condition of Theorem \ref{th1} into an optimization problem and derive a loss function that enforces this condition.

\noindent 
The dataset used for training as well as the loss function used in the training process are detailed in the following subsections. 

\subsection{Collocation point generation}
Since we do not rely on data points for the training, we generate collocation points and construct a physics dataset $D=\{\hat{x}^{(j)}, y^{(j)}\}$ to minimize the physics loss by uniformly sampling $(\mathcal{X},\mathcal{Y})$, with $j=1,...,N$ representing the sample number and $N$ the maximum number of collocation points.   

\subsection{Construction of the loss function}
Instead of minimizing a loss function associated with data trajectories plus a physics term, we focus exclusively on a physics-based loss function. The loss function is formulated to penalize deviations from the contraction condition across the domain of interest, and is a combination of the contraction condition and the boundary condition loss functions:
\begin{equation}
\mathcal{L} (\hat{x}^{(j)},y^{(j)})= \mu_1 \mathcal{L}_{\text{MPDI}} (\hat{x}^{(j)},y^{(j)}) + \mu_2 \mathcal{L}_{\text{BC}} (\hat{x}^{(j)}),
\label{loss}
\end{equation}

where $\mathcal{L}_{\text{MPDI}} $ and $\mathcal{L}_{\text{BC}}$ represent respectively the contraction condition loss function and the boundary condition loss function, and $\mu_1$ and $\mu_2$ are weighting coefficients.

\subsubsection{Contraction loss} The primary loss function is designed to enforce the contraction condition, ensuring that the system's differential dynamics satisfies the required inequality \eqref{contractth} which becomes
\begin{equation}
    \label{eq:contractionEquality}
   \underbrace{ \operatorname{He}\left\{\frac{\partial f}{\partial \hat{x}}(\hat{x})+\frac{\partial \hat{k}_\theta}{\partial \hat{x}}(\hat{x}, y)\right\} +2\lambda I }_{D(\hat{x},y)}\leqslant 0.
\end{equation} 
The above Matrix Partial Differential Inequality (MDPI) is expressed in a negative semi-definiteness sense. Since $D(\hat{x},y)$ is symmetric, we exploit the properties of 
its principle minors to impose the negative semi-definiteness condition in \eqref{eq:contractionEquality} by recalling the following Sylvester's criterion.

\begin{theorem} \cite{gilbert1991positive}
\label{th2SylvCrit}
A symmetric matrix \( D\in\mathbb{R}^{n\times n},\) is negative semi-definite if and only if \( (-1)^i \Delta_i \geq 0 \) for all leading principal minors \( \Delta_i \), where \( i = 1, 2, \dots, n \).
\end{theorem}

\noindent
Therefore,  $\mathcal{L}_{\text{MPDI}}$ is constructed by penalizing the leading principal minors of \eqref{eq:contractionEquality} as follows
\begin{equation}
\label{eq:minors}
\mathcal{L}_{\text{MPDI}} (\hat{x}^{(j)},y^{(j)}) = \frac{1}{N} \sum_{j=1}^{N} \sum_{i=1}^{n} \rho_i l_i(\hat{x}^{(j)},y^{(j)}),
\end{equation}

where $\rho_i$ are weighting coefficients for the principle minors loss function $l_i$ given by 

\begin{equation*}
    l_i(\hat{x}^{(j)},y^{(j)})=\left\{\begin{array}{l c}
      \max(0, \Delta _i (\hat{x}^{(j)},y^{(j)}))   &; i \text{ is odd} \\
      \min(0, \Delta _i(\hat{x}^{(j)},y^{(j)}))   &; i \text{ is even}  
    \end{array}   \right.
\end{equation*}

with $\Delta _i(\hat{x}^{(j)},y^{(j)})$ representing the leading principle minors of $D(\hat{x}^{(j)},y^{(j)})$, for $i=1,...,n$.

\subsubsection{Boundary condition loss}
The loss function for the boundary condition is constructed by considering that if the second argument of the gain $\hat{k}_\theta$ is equal to the output function of its first argument, then the observer's trajectory is aligned with the trajectory of the system, and the observer's gain is identically zero. One way to enforce this condition is to consider the following loss function:
\begin{equation}
\mathcal{L}_{\text{BC}} (\hat{x}) = \frac{1}{N} \sum_{j=1}^{N} ||\hat{k}_\theta(\hat{x}^{(j)},h(\hat{x}^{(j)}))||^2. 
\label{eqBC}
\end{equation}

\subsection{Training Pseudo-Code}
Training the PINN requires computing the loss function in \eqref{loss}, which in turn requires computing the jacobian of $f$ and $\hat{k}_\theta$. Since the description of $f$ is given by the model, the Jacobian $J_x = \frac{\partial f}{\partial \hat{x}}$ is computed analytically to avoid any numerical differentiation errors. On the other hand,  the jacobian of the gain $\frac{\partial \hat{k}_\theta}{\partial \hat{x}}$ is computed numerically by exploiting the actual automatic differential frameworks behind neural networks optimization. The training algorithm is provided upon request. 

\begin{algorithm}
\caption{Training Algorithm of the PINN-based Contraction Nonlinear Observer}
\begin{algorithmic}[1]
    \State \textbf{Input:} Dataset $D$, initial parameters $( \theta )$ of $( T_{\theta}(\hat{x}, y) )$, $( f )$ from equation \eqref{model}, the Jacobian $( J_x = \frac{\partial f}{\partial \hat{x}} $), learning rates $( \alpha )$ (Adam), $( \beta $) (L-BFGS), number of epochs $( N_{\text{epochs}} $) and loss weights $( \mu_1, \mu_2 )$. 
    
    \State // \textbf{Train with Adam}
    \For{$( q = 1, \dots, N_{\text{epochs}} )$}
        \For{each batch $( B \subset D)$}
            \State Compute $( \hat{k}_{\theta}(\hat{x}, y) = T_\theta(\hat{x}, y) $) for $( (\hat{x}, y) \in B )$
            \State Compute $\frac{\partial \hat{k}_\theta}{\partial \hat{x}}(\hat{x}, y)$ numerically for $( (\hat{x}, y) \in B )$
            \State Compute $( J_x = \frac{\partial f }{\partial \hat{x}}(\hat{x},y) )$ for $( (\hat{x}, y) \in B )$
            \State Compute $( k_{\theta}(\hat{x}, h(\hat{x})) = T_\theta(\hat{x}, h(\hat{x})) $) for $( \hat{x} \in B )$
            
            \State Compute $(\mathcal{L} = \mu_1 \mathcal{L}_{\text{MPDI}} + \mu_2 \mathcal{L}_{\text{BC}} )$, where $( \mathcal{L}_{\text{MPDI}}, \mathcal{L}_{\text{BC}} )$ are given by equations \eqref{eq:minors} and \eqref{eqBC}
            
            \State Update $( \theta_{\text{new}} = \theta_{\text{old}} - \alpha \cdot \nabla_\theta \mathcal{L} )$
        \EndFor
    \EndFor
    
    \State // \textbf{Train with L-BFGS}
    \For{$( q = 1, \dots, N_{\text{epochs}} $)}
        \For{each batch \( B \subset D\)}
            \State repeat steps 5-9
            \State Update $( \theta_{\text{new}} = \theta_{\text{old}} - \beta \cdot \nabla_\theta \mathcal{L} $)
        \EndFor
    \EndFor
\end{algorithmic}
\end{algorithm}

\section{Robustness to neural network approximation error and measurement noise}
\label{sec:robustness}
In this section, we study the effect of the neural network approximation error and the measurement noise on the convergence of the observer by considering the following structure 

\begin{equation}
\label{observerPINN}
\left\{\begin{array}{l}
\dot{\hat{x}}=f(\hat{x})+\hat{k}_\theta(\hat{x}, y_e)
\hat{y}=h(\hat{x}),
\end{array}\right.
\end{equation}

with $y_e$ and $\hat{k}_{\theta}$ representing respectively the noisy output and learned observer gain, which is given by
\begin{equation}
y_e=h(x)+v(t),
\label{noisyoutput}
\end{equation}
\begin{equation}
\hat{k}_\theta(\hat{x}, y) =k(\hat{x}, y)-\varepsilon(\hat{x}, y),
\label{apperror}
\end{equation}
where $v$ is the measurement noise and $\varepsilon$ is the neural network approximation error.

To study the robustness of the proposed PINN-based contraction nonlinear observer, we consider the following assumptions:
 
\begin{ass} The learned observer gain $\hat{k}_\theta$ is Lipschitz on its second argument, uniformly in $\hat{x}$ 
\begin{equation}
\label{lipschitz}
 \norm{ \hat{k}_\theta(., y_1)-\hat{k}_\theta(., y_2)}\leq  L \norm{y_1-y_2},   
\end{equation}
with $L>0$
\label{assumlipschitz}
\end{ass}

\begin{ass} There exists two positive and bounded constants $\bar{v}$ and $\bar{\epsilon}$ such that the measurement noise $v(t)$ and the neural network approximation error $\varepsilon$ satisfy 
\begin{itemize}
    \item (A1):  $\norm{v(t)}\leqslant \Bar{v} \quad \forall t\in \mathbb{R}$ 
    \item (A2):   $\norm{ \varepsilon(\hat{x}, y)}\leqslant \Bar{\varepsilon} \quad \forall (\hat{x}, y)\in \mathcal{X}\times \mathcal{Y}$.   
\end{itemize}
\label{assumvande}
\end{ass}

\noindent
Assumption \ref{assumlipschitz} can be easily satisfied by selecting a Lipschitz activation function for the neural network $T_\theta$. Moreover, since the states and the output of the system are bounded as per Assumption \ref{assumption1}, the approximation error $\varepsilon$ is bounded. 

\begin{theorem}\label{thISS}
Consider system \eqref{model} with the noisy output in \eqref{noisyoutput}, and the observer in \eqref{observerPINN}. Let assumptions \ref{assumlipschitz} and \ref{assumvande} hold. Let $\lambda$ be the contraction rate in \eqref{contractth} for $P=I$. If there exists a strictly positive constant $\eta$ and  $\lambda>2$ then the estimation error is exponentially input to state stable and satisfies:
\begin{equation}
\|x(t)-\hat{x}(t)\| \leqslant \|x(0)-\hat{x}(0)\| e^{-\eta t}+\frac{1}{2 \eta^{\frac{1}{2}}} \bar{\varepsilon}+\frac{L}{2 \eta^{\frac{1}{2}}} \bar{v}.\\
\label{estimationerror}
\end{equation}

\end{theorem}

\begin{proof}
Consider the following Lyapunov function\\

$ V=\frac{1}{2}\norm{x-\hat{x}}^2$

\begin{align}
\dot{V}& =(x-\hat{x})^T \left[f(x)-f(\hat{x})-\hat{k}_\theta\left(\hat{x}, y_e\right)\right]  \notag \\ 
& =(x-\hat{x})^T \left[f(x)-f(\hat{x})-\hat{k}_\theta \left(\hat{x}, y_e\right)+\hat{k}_\theta (\hat{x}, y)-\hat{k}_\theta (\hat{x}, y)\right] \notag \\
& =(x-\hat{x})^T [f(x)-f(\hat{x})-k(\hat{x}, y)+\varepsilon(\hat{x}, y) \notag \\
& \hspace{1.7cm}\left.+\hat{k}_\theta(\hat{x}, y)-\hat{k}_\theta\left(\hat{x}, y_e\right)\right] \notag \\
& =(x-\hat{x})^T [f(x)-f(\hat{x})-k(\hat{x}, y)] +(x-\hat{x})^T \varepsilon(\hat{x}, y) \notag \\
& +(x-\hat{x})^T \left[\hat{k}_\theta(x, y)-\hat{k}_\theta\left(\hat{x}, y_e\right)\right]. \label{eqV1}
\end{align}

Using Young's Inequality, \eqref{eqV1} becomes 

\begin{align} 
\dot{V} \leqslant & (x-\hat{x})^{T} [f(x)-f(\hat{x})-k(\hat{x}, y)]  +||x-\hat{x}||^2  \notag \\
& +\frac{1}{2}\|\varepsilon(\hat{x}, y)\|^2 +\frac{1}{2}\left\|\hat{k}_\theta(\hat{x}, y)-\hat{k}_\theta\left(\hat{x}, y_e\right)\right\|^2.
\label{eqV2}
\end{align}

Substituting \eqref{lipschitz} in \eqref{eqV2}, one obtains 
\begin{align}
\dot{V} \leqslant & (x-\hat{x})^{T} [f(x)-f(\hat{x})-k(\hat{x}, y)]  +||x-\hat{x}||^2  \notag \\
& +\frac{1}{2}\|\varepsilon(\hat{x}, y)\|^2 +\frac{L^2}{2}\left\|v(t)\right\|^2. \label{eqV3}
\end{align}
 Using the identity in \cite[see Theorem 4.3 p.~231] {BERNARD2022} for $y=h(x)$, one obtains 
\begin{equation}
  \dot{V} \leqslant -(\lambda -2) V+\frac{1}{2}\|\varepsilon\|^2+\frac{L^2}{2} \|v\|^2 .
\end{equation}

Finally, the exponential input to state stability is ensured if $\lambda > 2$ and the estimation error is given by \eqref{estimationerror} for $\eta=\lambda -2$.
\end{proof}

\noindent
Theorem \ref{thISS} indicates that increasing the contraction rate $\lambda$ can further mitigate the impact of the neural network approximation errors and measurement noise on the estimation. However, a higher contraction rate reduces the likelihood of finding a gain $k$ that satisfies the MPDI.

\section{Numerical Simulations}\label{sec:sim}

To evaluate the performance and disturbance and noise rejection of the proposed learning-based contraction nonlinear observer, we perform numerical simulations on two nonlinear systems: nonlinear Van der Pol and reverse Duffing oscillators given by \eqref{VdP} and \eqref{duffing}, respectively. 
 
\begin{tabular}{p{5.2cm} p{3cm}}
 \begin{equation}\label{VdP}
\!\!\!\!\left\{\begin{array}{l}
\dot{x}_{1}\!=  x_{2} \\
\dot{x}_{2}=-x_{1}+x_2(1-x_{1}^2)\! \\
y=x_{1}
\end{array}\right.
\end{equation}  &  \begin{equation}\label{duffing}
    \!\!\!\!\left\{\begin{array}{l}
    \dot{x}_{1}\!=  x_{2}^3 \\
    \dot{x}_{2}= -x_{1}\! \\
    y = x_{1}    
\end{array}\right.
\end{equation}
\end{tabular}

\begin{figure}
    \centering
    \includegraphics[width=1\linewidth]{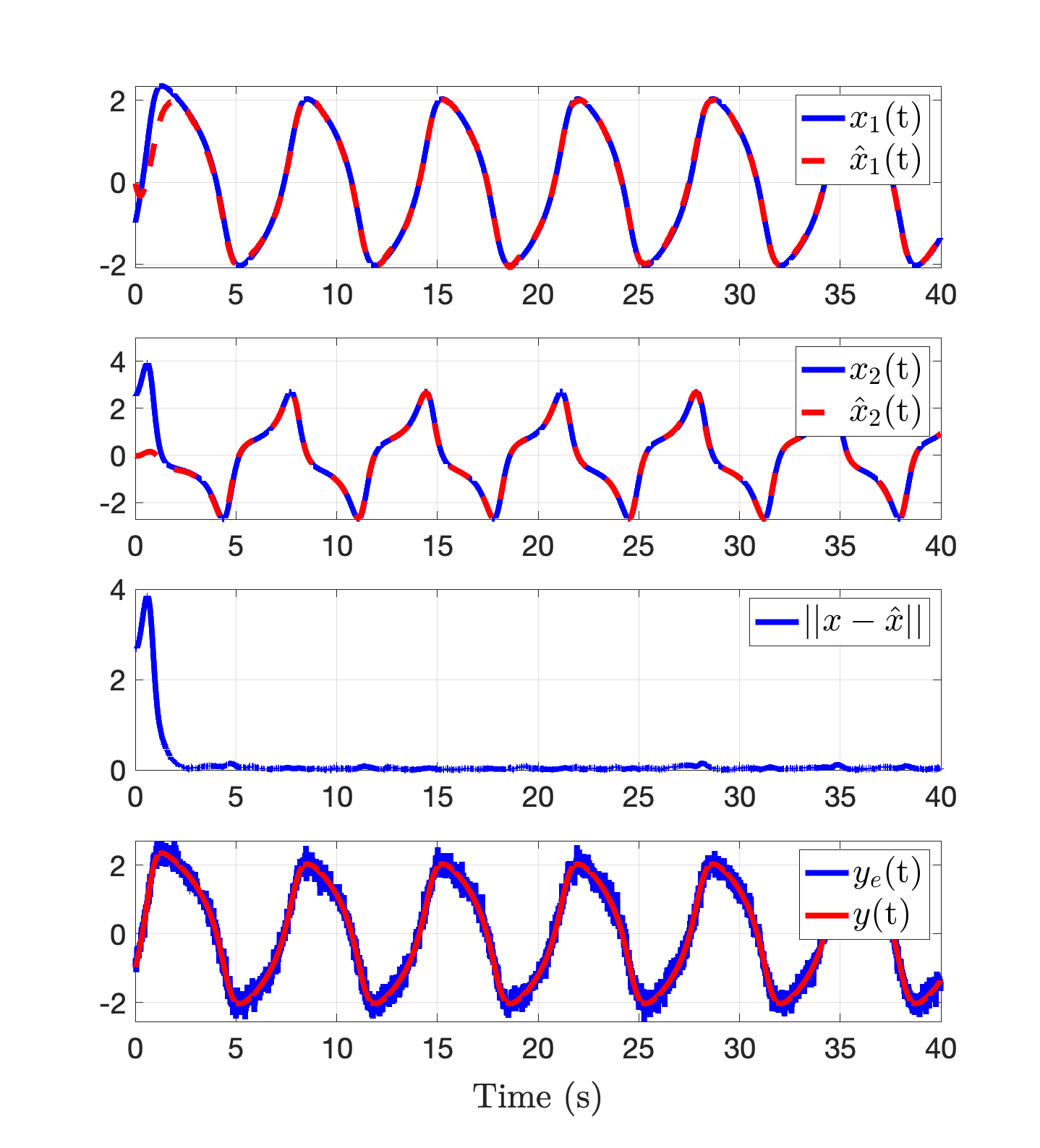}
    \caption{State estimation of the Van der Pol Oscillator under 15$\%$ of measurement noise using the proposed PINN-based contraction observer.}
    \label{fig:VanderPolWnoise}
\end{figure}

\begin{figure}
    \centering
    \includegraphics[width=1\linewidth]{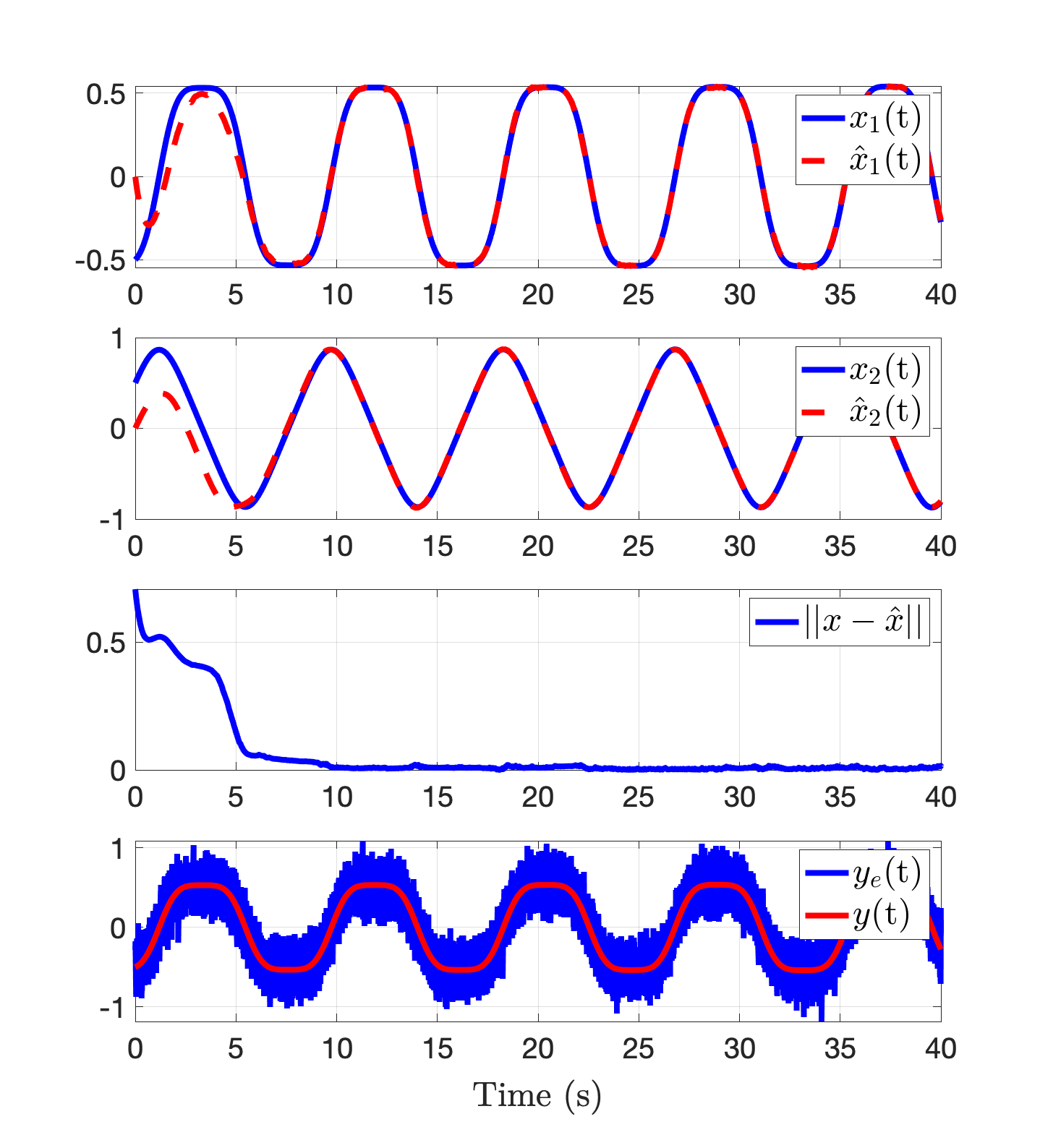}
    \caption{State estimation of the Inverse Duffing Oscillator under 15$\%$ of measurement noise using the proposed PINN-based contraction observer}
    \label{fig:DuffungWnoise}
\end{figure}

A training dataset of 4000 samples was generated by uniformly sampling the region of interest $(\mathcal{X}, \mathcal{Y})$, where $(\mathcal{X}, \mathcal{Y})=([-2,2] \times[-3,3], [-2, 2])$ for the Van der Pol oscillator and $(\mathcal{X}, \mathcal{Y})=([-1,1]^2, [-1, 1])$ for the reverse duffing. The considered architecture for both systems is a multi-layer perceptron of 5 hidden layers with 30 neurons each. The neural network was trained using the Adam optimizer \cite{kingma2014adam} with a learning rate of \( \alpha= 10^{-3}\) for 500 iterations followed by the same amount of iterations for the L-BFGS optimization algorithm \cite{byrd1995limited} with $\beta=1$, which aligns with the methodologies reported in the literature that show better performance for training physics-based deep learning models \cite{raissi2019, rathore2024challenges, lou2021physics, lehmann2023mixed, haitsiukevich2023improved,
antonelo2024physics, ekeland2024physics
}. The training was performed on a single NVIDIA A40 GPU. 

A contraction rate $\lambda$ of 2.5 is considered for both systems and the loss functions' weights for both systems are provided in Table \ref{tab:weights}. \\
\begin{table}
    \centering
    \caption{Loss functions' weights}
    \begin{tabular}{|p{2cm}|p{1cm}|p{1cm}|p{1cm}|p{1cm}|}
    \hline
      Weights  & $\mu_1$ & $\mu_2$ & $\rho_1$ & $\rho_2$  \\
      \hline
      Van der Pol & $10^{-3}$ & 1 & 1 & $10^{-1}$\\
      \hline
      Reverse Duffing &  1 & 1& 1&1\\
      \hline
    \end{tabular}
    \label{tab:weights}
\end{table}

\noindent
The simulation results are depicted in Fig.\ref{fig:VanderPolWnoise} and \ref{fig:DuffungWnoise} under a white Gaussian noise of $0.15$ magnitude. The Van der Pol oscillator was initialized at [-1, 2.5], the reverse duffing was initialized at [-0.5; 0.5], and the observer for both systems was initialized at the origin. One can see that the proposed learning-based contraction nonlinear observer was able to accurately estimate the states of both systems \eqref{VdP} and \eqref{duffing} and demonstrated good noise rejection. Furthermore, one can see that the estimation error stays bounded within a neighborhood of the origin defined by the approximation error of the neural network and the noise level. To further assess the robustness of the proposed observer to measurement noise, the PINN was trained for several values of the contraction rate $\lambda$, and the observer was simulated under the same level of measurement noise for $\hat{x}(0)=x(0)$. The results in Table $\ref{tab:noise}$ confirm the findings of Theorem \ref{thISS} and show that increasing the contraction rate reduces the effect of the measurement noise on the estimation error.  \\

\begin{table}
    \centering
    \caption{Mean squared estimation error in $\%$ for different contraction rates and 15$\%$ of measurement noise  }
    \begin{tabular}{|p{2cm}|p{1cm}|p{1cm}|p{1cm}|}
    \hline
      $\lambda$ & $2.5$ & $4$ & $5$   \\
      \hline
      Van der Pol & $0.53$ & $0.47$ & $0.22$ \\
      \hline
      Reverse Duffing & $0.062$  & $0.016$ & $0.013$ \\
      \hline
    \end{tabular}
    \label{tab:noise}
\end{table}

\section{Conclusion}\label{sec:conclusion}
The present paper addressed a long-lasting drawback of contraction-based nonlinear observer design, by proposing an unsupervised learning-based approach to design the nonlinear observer's gain. The proposed approach relies on a physics-informed neural network that enforces the contraction condition to better approximate the gain of an exponentially stable observer for an autonomous nonlinear system. The effect of the neural network approximation error and measurement noise was studied, and conditions for ensuring exponential input-to-state stability were derived. The proposed learning-based contraction nonlinear observer demonstrated good performance and noise rejection in numerical simulation. Future work focuses on extending the proposed approach to nonlinear non-autonomous systems.

\section*{Acknowledgment}
\noindent Research reported in this publication was supported by King Abdullah University of Science and Technology (KAUST) with the  Base Research Fund (BAS/1/1627-01-01) and (BAS/1/1665-01-01), and the National Institute for Research in Digital Science and Technology (INRIA). The authors would also like to thank Ibrahima N'doye for the fruitful discussions that led to this paper.

\bibliographystyle{unsrt}
\bibliography{arxivTex}
\end{document}